\documentclass[conference]{IEEEtran}
\IEEEoverridecommandlockouts

\usepackage{verbatim}
\usepackage{epsfig}
\usepackage{amsmath}
\usepackage{amsthm}
\usepackage{amssymb,graphicx}
\usepackage{psfrag}
\usepackage[T1]{fontenc}
\usepackage[ruled]{algorithm2e}
\usepackage{mathtools, nccmath}
\usepackage{cuted}
\usepackage{multirow}
\usepackage{etoolbox}
\AfterEndEnvironment{strip}{\leavevmode}
\usepackage{url}
\usepackage{dsfont}
\usepackage{mathtools}

\usepackage[usenames,dvipsnames]{pstricks}
\usepackage[ruled]{algorithm2e}
\usepackage{etoolbox}
\usepackage{graphicx}
\graphicspath{{Images/}}
\ifCLASSOPTIONcompsoc
\usepackage[caption=false,font=normalsize,labelfont=sf,textfont=sf]{subfig}
\else
\usepackage[caption=false,font=footnotesize]{subfig}
\fi
\usepackage{stfloats}
\makeatletter
\patchcmd{\@makecaption}
  {\scshape}
  {}
  {}
  {}
\makeatother
\usepackage[noadjust]{cite}

\usepackage[font=footnotesize,skip=0pt]{caption}

\newcommand{\df}{$\mathrm{detection\_flag}\ $}

\newtheorem{theorem}{Theorem}[]

\newtheorem{assump}{Assumption}

\newtheorem{remark}{Remark}

\definecolor{color1}{rgb}{0,0,0}
\usepackage{textcomp}
\usepackage{xcolor}
\def\BibTeX{{\rm B\kern-.05em{\sc i\kern-.025em b}\kern-.08em
    T\kern-.1667em\lower.7ex\hbox{E}\kern-.125emX}}
\IEEEpeerreviewmaketitle 
\begin{document}
\title{A Distributed Malicious Agent Detection Scheme for Resilient Power Apportioning in Microgrids 
\thanks{This research is supported by the United States Department of Energy via grant number DE-CR$0000040$.
}
}
\author{\IEEEauthorblockN{Vivek Khatana$^{1}$, Soham Chakraborty$^{2}$, Govind Saraswat$^{3}$, Sourav Patel$^{4}$, and Murti V. Salapaka$^{5}$}
\IEEEauthorblockA{Department of Electrical and Computer Engineering, University of Minnesota, USA
\\ \{$^1$khata010, $^2$chakr138, $^3$saras006,  $^5$murtis\}@umn.edu, \{$^4$patel292@alumni.umn.edu\} }}
\maketitle

\begin{abstract}
We consider the framework of distributed aggregation of Distributed Energy Resources (DERs) in power networks to provide ancillary services to the power grid. Existing aggregation schemes work under the assumption of trust and honest behavior of the DERs and can suffer when that is not the case. In this article, we develop a distributed detection scheme that allows the DERs to detect and isolate the maliciously behaving DERs. We propose a model for the maliciously behaving DERs and show that the proposed distributed scheme leads to the detection of the malicious DERs. Further, augmented with the distributed power apportioning algorithm the proposed scheme provides a framework for \textit{resilient} distributed power apportioning for ancillary service dispatch in power networks. A controller-hardware-in-the-loop (CHIL) experimental setup is developed to evaluate the performance of the proposed resilient distributed power apportioning scheme on an 8-commercial building distribution network (Central Core) connected to a 55 bus distribution network (External Power Network) based on the University of Minnesota Campus. A diversity of DERs and loads are included in the network to generalize the applicability of the framework. The experimental results corroborate the efficacy 
of the proposed resilient distributed power apportioning for ancillary service dispatch in power networks.\\[1ex]
\end{abstract}
\IEEEpeerreviewmaketitle
\section{Introduction}

With the proliferation of power electronic converter technologies, renewable energy resources, and battery storage systems, the modern microgrid network is transitioning towards integrating a large number of smaller distributed energy resources (DERs) scattered throughout the network. Further, with the presence of the DERs the dispatch methodologies for deploying ancillary services
for reliable operation of the modern grid are changing
significantly (for example, the Electric Reliability Council of Texas’s (ERCOT’s) responsive reserve
services (RRS)), thereby demonstrating the viability of widescale adoption of DERs for such applications \cite{ercot_1,ercot_2}. To achieve grid ancillary services a large number of DERs must be coordinated on a fast timescale. The coordination of multiple DERs presents significant challenges. Several articles (see \cite{unamuno2015hybrid,xi2018power} and references therein) in the literature have introduced centralized approaches wherein a centralized controller at the distribution level sends dispatch commands to local actuators and requires information from all of the DERs. Such centralized approaches lack flexibility 
and scalability, and they require expensive
high-performance computing and high-speed communication
networks to meet the ancillary service requirements satisfactorily. To mitigate these challenges, distributed control approaches are proposed with DERs as a multi-agent system \cite{replace_1, replace_3, replace_2}.  Advantages of distributed approaches include coordination using only local computations and plug-and-play capability. Article \cite{dorfler2017gather} presents a gather-broadcast method, a distributed average
integral method was developed in \cite{zhao2015distributed}; however, these methods are often sensitive 
to gain coefficients and can result in slow convergence to
dispatch outputs. Article \cite{megel2017distributed} develops a distributed approach but does not provide any theoretical guarantees and relies on heuristics to achieve dispatch requests. A distributed consensus-based power apportioning algorithm overcoming the limitations of the above methods is developed in \cite{apportioning1, apportioning2}. We focus on the distributed power apportioning algorithms proposed in articles \cite{apportioning1, apportioning2}.

\subsection{Literature Review}
A key component of the distributed apportioning schemes in \cite{apportioning1, apportioning2} is achieving
consensus between the decision variables of different DERs. The commonly employed consensus algorithms in the literature \cite{apportioning1, apportioning2,switching_top, ADMM_tac} assume an underlying trust and credibility of the participating agents. When factoring in the distributed coordination for ancillary service dispatch, any DER's lack of honesty can skew the resulting decision, which can adversely affect the power network. Thus, it is imperative to employ resilient measures to counter the presence of malicious agents in the network. Several works in the literature have studied the problem of securing the consensus updates \cite{marsh_formalising, toulouse_defense, sundaram_robust, panagou_resilient, intruder_codit}. A trust-based model for multi-agent systems is considered in \cite{marsh_formalising}. The scheme requires repeated interactions between the agents and necessitates gathering a large amount of information from different agents in the network to formulate notions of cooperation or rejection of information. Article \cite{toulouse_defense} devises an intrusion detection system based on network traffic analysis. In contrast, authors in \cite{sundaram_robust, panagou_resilient} assert designing the interaction topology between the agents to achieve an agreement, that is difficult to generalize. Article \cite{intruder_codit} develops a scheme based on the
monotonic properties of global maximum and minimum states of the agents under the consensus algorithm to
detect intrusions. However, the method doesn't isolate the malicious agents. \\
\hspace*{0.1in} Existing distributed schemes for aggregating DERs towards providing ancillary services do not prescribe a resilient operation in the presence of maliciously behaving DERs. Here, we propose a novel distributed intruder detection scheme that enables honest DERs to detect and isolate a maliciously behaving DER in the distributed apportioning algorithm. 

\subsection{Contributions}
\noindent The major contributions of this article are:\\
$1)$ We develop a distributed intruder detection scheme for distributed power apportioning with DERs in microgrids.\\
$2)$ The developed scheme allows the DERs to locally detect and isolate the communication links of the maliciously behaving DERs in their neighborhood. \\
$3)$ Compared to the existing works in the literature the proposed scheme leads to distributed detection and isolation of malicious DERs and is applicable for general interaction topologies. \\
$4)$ The developed scheme has a low computation and communication footprint unlike existing schemes in the literature \cite{sundaram_robust, panagou_resilient} that require a polynomial number of computations and communication in the number of DERs. The proposed algorithm works only in the neighborhood of a DER and thus, scales well due to local interactions. \\

The rest of the article is organized as follows: The system description and details of the distributed power apportioning setup are presented in Section~\ref{sec:system_desc}. Section~\ref{sec:malicious_agent_desc} presents the intrusion model and details the proposed resilient distributed apportioning algorithm. A controller hardware-in-the-loop (CHIL)-based experiment is conducted to evaluate the efficacy of the proposed intruder detection scheme in Section~\ref{sec:results}. In particular, we augment the developed intruder detection algorithm with the distributed power apportioning engine in \cite{apportioning1,apportioning2} towards meeting the ancillary service demand of a microgrid with critical infrastructure. The microgrid is emulated based on the University of Minnesota Twin Cities campus, consisting of $8$ commercial-scale buildings and an external $55$ bus distribution network, with DERs running the power apportioning algorithm, based on the suburb of Minneapolis. (The emulated power network is described in more detail in Section~\ref{sec:microgrid}). The laboratory test results establish the efficacy of the proposed scheme towards resilient distributed apportioning for ancillary service dispatch. Concluding remarks and future direction of research are provided in Section~\ref{sec:conclusion}.

\section{Description of the System}\label{sec:system_desc}
\subsection{Description of the Microgrid}\label{sec:microgrid}
\begin{figure*}[t]
	\centering
	\subfloat[]{\includegraphics[scale=0.505,trim={0cm 0cm 7cm 5cm},clip]{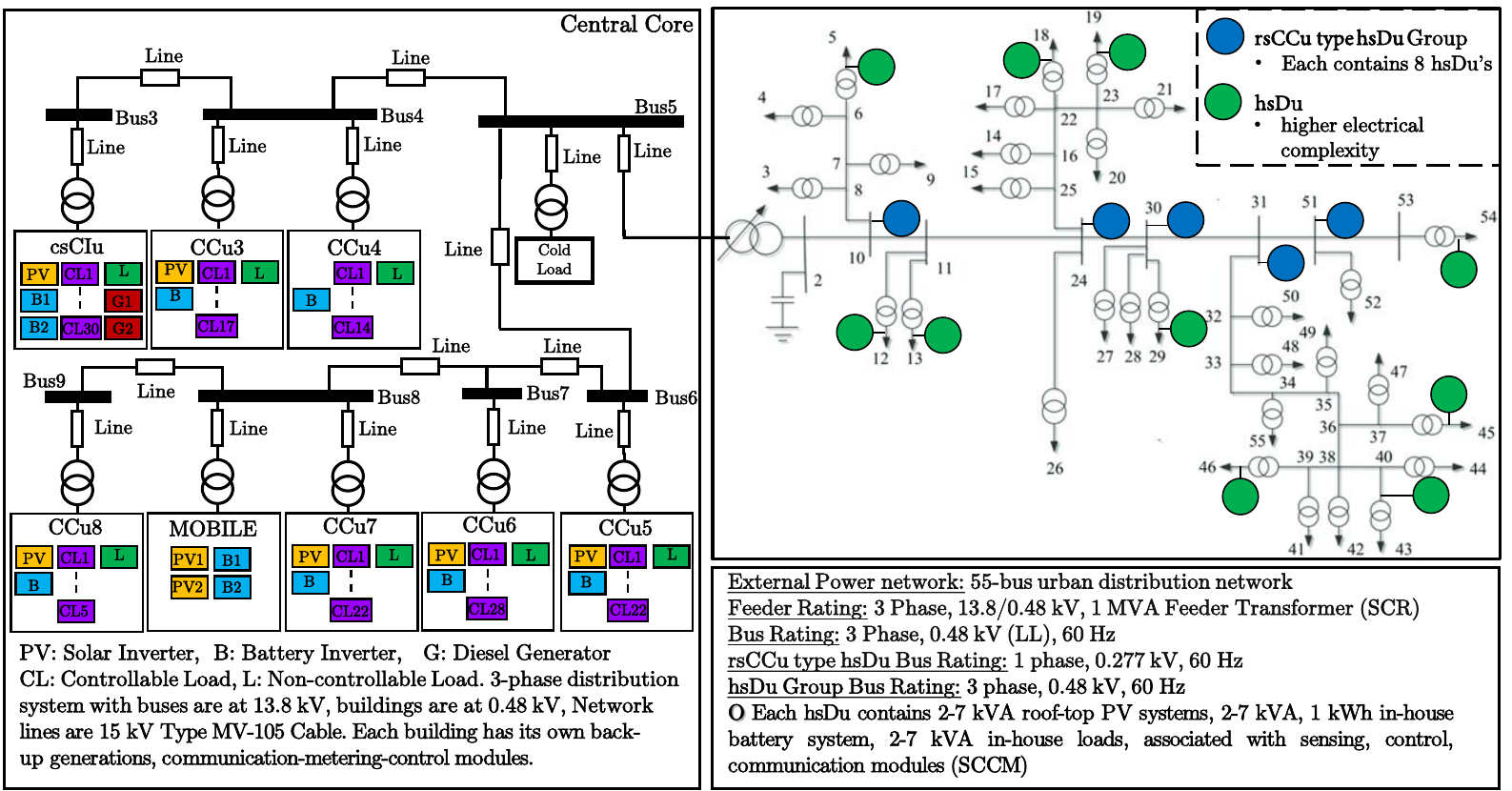}%
	\label{fig:microgrid}}\hspace{-0.05in}
	\subfloat[]{\includegraphics[scale=0.37,trim={0cm 0cm 23cm 0cm},clip]{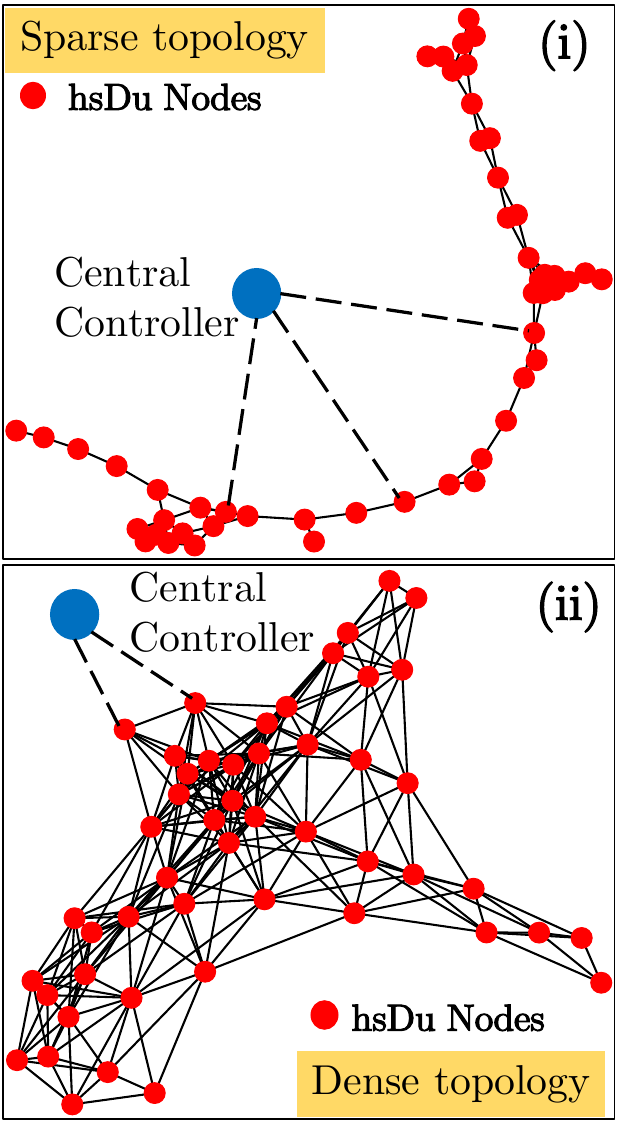}%
	\label{fig:communication}}
	\caption{Figure of (a) power network of the microgrid under study, (b) communication topology used in the power apportioning.}
	\label{fig:microgrid_communication}
\end{figure*}

\begin{table*}[b]
\centering
\caption{RATINGS OF VARIOUS SOURCES AND LOADS IN CENTRAL CORE OF THE MICROGRID}
\label{tab:rating1}
\resizebox{\linewidth}{!}{
\scriptsize
\begin{tabular}{c|ccccccccc}
\hline
\multirow{2}{*}{$\mathbf{Types}$} & \multicolumn{9}{c}{$\mathbf{Central~Core~of~the~Microgrid}$} \\ \cline{2-10} 
 & \multicolumn{1}{c|}{$\mathbf{csCIu}$} & \multicolumn{1}{c|}{$\mathbf{CCu3}$} & \multicolumn{1}{c|}{$\mathbf{CCu4}$} & \multicolumn{1}{c|}{$\mathbf{CCu5}$} & \multicolumn{1}{c|}{$\mathbf{CCu6}$} & \multicolumn{1}{c|}{$\mathbf{CCu7}$} & \multicolumn{1}{c|}{$\mathbf{CCu8}$} & \multicolumn{1}{c|}{$\mathbf{Mobile~Unit}$} & $\mathbf{Cold~Load}$ \\ \hline
\begin{tabular}[c]{@{}c@{}}$\mathrm{Battery~GFM~Inverter~(kVA)}$\end{tabular} & \multicolumn{1}{c|}{$200$} & \multicolumn{1}{c|}{$25$} & \multicolumn{1}{c|}{$100$} & \multicolumn{1}{c|}{$80$} & \multicolumn{1}{c|}{$200$} & \multicolumn{1}{c|}{$100$} & \multicolumn{1}{c|}{$50$} & \multicolumn{1}{c|}{$100$} & - \\ \hline
\begin{tabular}[c]{@{}c@{}}$\mathrm{Diesel~Gen-Set~(kVA)}$\end{tabular} & \multicolumn{1}{c|}{$100$} & \multicolumn{1}{c|}{-} & \multicolumn{1}{c|}{-} & \multicolumn{1}{c|}{-} & \multicolumn{1}{c|}{-} & \multicolumn{1}{c|}{-} & \multicolumn{1}{c|}{-} & \multicolumn{1}{c|}{-} & - \\ \hline
\begin{tabular}[c]{@{}c@{}}$\mathrm{PV~GFL~Inverter~(kVA)}$\end{tabular} & \multicolumn{1}{c|}{$50$} & \multicolumn{1}{c|}{$50$} & \multicolumn{1}{c|}{-} & \multicolumn{1}{c|}{$100$} & \multicolumn{1}{c|}{$100$} & \multicolumn{1}{c|}{$50$} & \multicolumn{1}{c|}{$100$} & \multicolumn{1}{c|}{$50$} & - \\ \hline
\begin{tabular}[c]{@{}c@{}}$\mathrm{Uncontrollable~Load~(kVA)}$\end{tabular} & \multicolumn{1}{c|}{$250$} & \multicolumn{1}{c|}{$30$} & \multicolumn{1}{c|}{$35$} & \multicolumn{1}{c|}{$30$} & \multicolumn{1}{c|}{$40$} & \multicolumn{1}{c|}{$25$} & \multicolumn{1}{c|}{$30$} & \multicolumn{1}{c|}{-} & $200$ \\ \hline
\begin{tabular}[c]{@{}c@{}}$\mathrm{Controllable~Load~(kVA)}$\end{tabular} & \multicolumn{1}{c|}{$500$} & \multicolumn{1}{c|}{$300$} & \multicolumn{1}{c|}{$315$} & \multicolumn{1}{c|}{$300$} & \multicolumn{1}{c|}{$360$} & \multicolumn{1}{c|}{$225$} & \multicolumn{1}{c|}{$270$} & \multicolumn{1}{c|}{-} & - \\ \hline
\end{tabular}}
\end{table*}

\begin{table*}[t]
\centering
\caption{RATINGS OF VARIOUS SOURCES AND LOADS IN EXTERNAL NETWORK OF THE MICROGRID}
\label{tab:rating2}
\resizebox{\linewidth}{!}{%
\scriptsize

\begin{tabular}{c|cccccccccc|ccccc}
\hline
\multirow{2}{*}{\begin{tabular}[c]{@{}c@{}}$\mathbf{Types}$\\ $\mathbf{(kVA)}$\end{tabular}} & \multicolumn{10}{c|}{$\mathbf{hsDu~in~External~Microgrid}$} & \multicolumn{5}{c}{$\mathbf{hsDu~Groups}$} \\ \cline{2-16} 
 & \multicolumn{1}{c|}{\begin{tabular}[c]{@{}c@{}}$\mathrm{No.~1}$ \end{tabular}} & \multicolumn{1}{c|}{\begin{tabular}[c]{@{}c@{}}$\mathrm{No.~2}$\end{tabular}} & \multicolumn{1}{c|}{\begin{tabular}[c]{@{}c@{}}$\mathrm{No.~3}$ \end{tabular}} & \multicolumn{1}{c|}{\begin{tabular}[c]{@{}c@{}}$\mathrm{No.~4}$ \end{tabular}} & \multicolumn{1}{c|}{\begin{tabular}[c]{@{}c@{}}$\mathrm{No.~5}$\end{tabular}} & \multicolumn{1}{c|}{\begin{tabular}[c]{@{}c@{}}$\mathrm{No.~6}$\end{tabular}} & \multicolumn{1}{c|}{\begin{tabular}[c]{@{}c@{}}$\mathrm{No.7}$\end{tabular}} & \multicolumn{1}{c|}{\begin{tabular}[c]{@{}c@{}}$\mathrm{No.~8}$\end{tabular}} & \multicolumn{1}{c|}{\begin{tabular}[c]{@{}c@{}}$\mathrm{No.9}$\end{tabular}} & \begin{tabular}[c]{@{}c@{}}$\mathrm{No.10}$\end{tabular} & \multicolumn{1}{c|}{\begin{tabular}[c]{@{}c@{}}$\mathrm{No.~1}$\end{tabular}} & \multicolumn{1}{c|}{\begin{tabular}[c]{@{}c@{}}$\mathrm{No.~2}$\end{tabular}} & \multicolumn{1}{c|}{\begin{tabular}[c]{@{}c@{}}$\mathrm{No.~3}$\end{tabular}} & \multicolumn{1}{c|}{\begin{tabular}[c]{@{}c@{}}$\mathrm{No.~4}$\end{tabular}} & \begin{tabular}[c]{@{}c@{}}$\mathrm{No.~5}$\end{tabular} \\ \hline
$\mathrm{PV}$ & \multicolumn{1}{c|}{$3$} & \multicolumn{1}{c|}{$5$} & \multicolumn{1}{c|}{-} & \multicolumn{1}{c|}{$5$} & \multicolumn{1}{c|}{$5$} & \multicolumn{1}{c|}{$5$} & \multicolumn{1}{c|}{$2$} & \multicolumn{1}{c|}{$5$} & \multicolumn{1}{c|}{$2$} & $2$ & \multicolumn{1}{c|}{$40$} & \multicolumn{1}{c|}{$24$} & \multicolumn{1}{c|}{$64$} & \multicolumn{1}{c|}{$32$} & $40$ \\ \hline
$\mathrm{Battery}$ & \multicolumn{1}{c|}{$3$} & \multicolumn{1}{c|}{$2$} & \multicolumn{1}{c|}{$8$} & \multicolumn{1}{c|}{-} & \multicolumn{1}{c|}{$5$} & \multicolumn{1}{c|}{-} & \multicolumn{1}{c|}{$7$} & \multicolumn{1}{c|}{$7$} & \multicolumn{1}{c|}{$2$} & $5$ & \multicolumn{1}{c|}{$16$} & \multicolumn{1}{c|}{$32$} & \multicolumn{1}{c|}{$32$} & \multicolumn{1}{c|}{$24$} & $24$ \\ \hline
$\mathrm{Load}$ & \multicolumn{1}{c|}{$3$} & \multicolumn{1}{c|}{$2$} & \multicolumn{1}{c|}{$2$} & \multicolumn{1}{c|}{$4$} & \multicolumn{1}{c|}{$6$} & \multicolumn{1}{c|}{1} & \multicolumn{1}{c|}{$3$} & \multicolumn{1}{c|}{$7$} & \multicolumn{1}{c|}{$3$} & $3$ & \multicolumn{1}{c|}{$24$} & \multicolumn{1}{c|}{$20$} & \multicolumn{1}{c|}{$24$} & \multicolumn{1}{c|}{$20$} & $28$ \\ \hline
\end{tabular}}
\end{table*}
The electrical network is developed based on a power network at the University of Minnesota campus, called the ``Central Core'', and the suburb Minneapolis distribution network. The schematic diagram of the electrical network is shown in Fig.~\ref{fig:microgrid_communication}\subref{fig:microgrid}, which consists of commercial-scale critical infrastructure units (csCIu) on Bus $3$, central core units (CCus) Bus $4$-$9$, mobile auxiliary generation units on Bus $8$, and a cold load station located at Bus $5$. Battery grid-forming (GFM) inverters \cite{DER_control_Madureira} (denoted as B,B$_1$,B$_2$ in Fig.~\ref{fig:microgrid_communication}\subref{fig:microgrid}) are designed based on $\mathrm{P}$-$\mathrm{f}$ and $\mathrm{Q}$-$\mathrm{V}$ droop law. Photovoltaic (PV) grid-following (GFL) inverter system \cite{DER_control_Madureira} is designed based on outer $\mathrm{P}$-$\mathrm{Q}$ controller. Diesel generation sets (denoted as G$_1$,G$_2$ in Fig.~\ref{fig:microgrid_communication}\subref{fig:microgrid}) are designed based on load-frequency-based governor control and $\mathrm{Q}$-$\mathrm{V}$-based automatic voltage regulator. The ratings of the central core are summarized in Table~\ref{tab:rating1}. The central core demands, $\rho_d(t_s)$, and since it contains the critical infrastructure, the external power network needs to support the central core and generate active power output that is equal to $\rho_d(t_s)$. The external power network is connected to the central core at the feeder near Bus-$5$. The power network is a $3$-$\phi$, $0.48$ kV, $60$ Hz, $55$-bus distribution network with $10$ residential apartments, called highly-scaled distributed units (hsDu), and $5$ residential-scale central core units (rsCCu), each of which contains $8$ hsDu. The external network has $50$ hsDus with controllable DERs. \textit{In the rest of the article, we will use the terms hsDu and DER interchangeably}. The ratings of the external power network are summarized in Table~\ref{tab:rating2}. To meet the requested power demand of the central core the external power network with $50$ DERs is employed with a power apportioning engine, running Algorithm~1 of \cite{apportioning1}. The central core, at a given time $t_s$, commands a total power demand set point based on its current situation. The DERs' decisions are aggregated via a communication network. Figs.~\ref{fig:microgrid_communication}(b)(i) and~(b)(ii) show two representative communication topologies, one with sparse connectivity and the other with dense connectivity respectively. Each node in the communication network is a hsDu with a communication interface, together referred to as an agent responsible for generating the active power set point for the corresponding DER based on the results obtained from the power apportioning engine. The power apportioning engine is described in detail next.

\subsection{Distributed Power Apportioning}\label{sec:dist_Appr}
The central controller is an entity interfacing
with the central core on one end and the hsDu network $\mathcal{G}=\{\mathcal{V}, \mathcal{E}\}$, with $| \mathcal{V}| = 50$, on the other. The central controller, similar to a power aggregator in the Transmission System, aggregates the DERs available in the external network to provide ancillary power to the central core. Considering the geographical span of the external power network it may not be feasible for all DERs to directly communicate with the central controller. Hence, only some DERs given by the set $\mathcal{N}_d \subseteq \mathcal{V}$, with $l=|\mathcal{N}_d|$ can directly communicate with the central controller. The central controller aggregates the central core's power demand $\rho_d(t_s)$, and relays it to DER $i, i\in\mathcal{N}_d$ at time instant $t_s \geq 0$. The DERs need to collectively meet the demand $\rho_d(t_s)$ while communicating only to their neighboring DERs in a connected communication graph $\mathcal{G}=\{\mathcal{V}, \mathcal{E}\}$ and respecting individual resource constraints. In particular, the DERs with power reference commands $\pi_i^*(t_s)$, ensure $\sum_{i\in \mathcal{V}}\pi_i^*(t_s) = \rho_d(t_s)$, and the capacity constraints of the DERs are not violated, that is $\pi_i^\mathrm{min}(t_s) \leq \pi_i^*(t_s) \leq \pi_i^\mathrm{max}(t_s)$, where $\pi_i^\mathrm{min}(t_s)$ and $\pi_i^\mathrm{max}(t_s)$ denote the maximum and minimum generation capacities of the DER $i$ at time instant $t_s$. State of the art on solving distributed apportioning involves an iterative finite-time average consensus protocol \cite{apportioning1,apportioning2}:\\
At each iteration $k$ of the protocol, every DER sends three internal states, $x_i(k), y_i(k), z_i(k)$ to its out-neighbors $\mathcal{N}_j^+:=\{i|(i,j)\in \mathcal{E}, i\neq j\}$, and receives these estimates from its in-neighbors $\mathcal{N}_j^-:=\{j|(j,i)\in\mathcal{E}, i\neq j\}$. The states are then updated as follows: 
\begin{align}
    x_i(k+1) &= \textstyle p_{ii}x_i(k)+\sum_{j\in \mathcal{N}_i^-}p_{ij}x_j(k), \label{eq:appr_update_num}\\
    y_i(k+1) &= \textstyle p_{ii}y_i(k)+\sum_{j\in \mathcal{N}_i^-}p_{ij}y_j(k), 
    \label{eq:appr_update_den} \\
    z_i(k+1) & = \textstyle \frac{1}{y_i(k+1)} x_i(k+1) \label{eq:appr_update_ratio}
\end{align}
where $p_{ij} > 0, i,j \in \mathcal{V}$ denotes predefined weights and $x_i(0) = \left(\rho_d(t_s)/l \right) - \pi_i^\mathrm{min}(t_s)$, if $i \in \mathcal{N}_d$,  $x_i(0) = - \pi_i^\mathrm{min}(t_s)$ if $i \in \mathcal{V} \setminus \mathcal{N}_d$, and $y_i(0) = \pi_i^\mathrm{max}(t_s) - \pi_i^\mathrm{min}(t_s)$. The updates~\eqref{eq:appr_update_num}-\eqref{eq:appr_update_ratio} are accompanied by additional maximum and minimum consensus protocols for terminating the iterations and producing a faster decision in finite time \cite{apportioning1,apportioning2}. In particular, when the difference between the global maximum and minimum goes below a specified tolerance $\varepsilon$ then algorithm iterations are terminated with the state value 
\begin{align}\label{eq:z_star}
    z_i^* = \textstyle \frac{\rho_d(t_s) - \sum_{i\in \mathcal{V}} \pi_i^\mathrm{min}(t_s)  }{\sum_{i\in \mathcal{V}} (\pi_i^\mathrm{max}(t_s) - \pi_i^\mathrm{min}(t_s)) }.
\end{align}
The DERs then dispatch the following power reference set-points, 
\begin{align*}
    \pi_i^*(t_s):=\pi_i^\mathrm{min}(t_s) & + z_i^*(\pi_i^\mathrm{max}(t_s)-\pi_i^\mathrm{min}(t_s)). \\
    \mbox{Note that} \ \textstyle \sum_{i\in \mathcal{V}} \pi_i^*(t_s) &= \rho_d(t_s).
\end{align*}
Articles \cite{apportioning1,apportioning2} provide more detail about the apportioning protocol.

\section{Malicious Behavior of DERs}\label{sec:malicious_agent_desc}
We considered the distributed apportioning network in the previous section and its ability to provide ancillary services to the central core with aggregation of hsDus. For the updates~\eqref{eq:appr_update_num}-\eqref{eq:appr_update_den} of any DER utilize the information shared by its neighbors under the assumption of trustworthy cooperation towards the objective of meeting the ancillary service goal for the central core. This leaves the DER network vulnerable to malicious agents trying to manipulate the power-apportioning protocol towards selfish goals or something even more severe by disrupting the convergence of the apportioning protocol to make the central core system unstable. We model the malicious behavior via an intruder model. Let $\mathcal{N}_m$ denote the set of maliciously behaving DERs. Let $k_m$ denote the iteration of the start of the malicious behavior for any malicious DER $m$. Any DER $m \in \mathcal{N}_m$ applies a linear drift $\theta_m(k)$ to its true state $x_m(k)$ and sends 
\begin{align}\label{eq:intruder_model}
    \hat{x}_m(k) := x_m(k) + \theta_m(k),
\end{align}
for all $k \geq k_m$, to its neighbors. The update of the neighboring DER $i$ of the malicious DER $m$ in the power apportioning protocol becomes,
\begin{align}
    x_i(k+1) &= \textstyle p_{ii}x_i(k)+\sum_{j\in \mathcal{N}_i^- \setminus m}p_{ij}x_j(k) \nonumber \\
    & \hspace{1in} + p_{im} (x_m(k) + \theta_m(k)). \label{eq:appr_update_num_mal}
\end{align}
Different choices of the perturbation $\theta_m$ allow the malicious DER $m$ to achieve different objectives, in particular
\begin{itemize}\label{list:attack_cases}
    \item [a)] if $\theta_m (k)$ is a divergent (bounded or unbounded) sequence then the updates~\eqref{eq:appr_update_num}-\eqref{eq:appr_update_den} for the DERs do not converge
    \item [b)] if $\theta_m (k)$ is convergent series with  $\sum_{k} \theta_m (k)  = \alpha$. This leads to the DER solution $z_i^*$ to converge to $\frac{\rho_d(t_s) - \sum_{i\in \mathcal{V}} \pi_i^\mathrm{min}(t_s) + \alpha  }{\sum_{i\in \mathcal{V}} (\pi_i^\mathrm{max}(t_s) - \pi_i^\mathrm{min}(t_s)) } $ (this can be shown using \cite{apportioning2}, Theorem~{III.1}), leading to $\sum_{i\in \mathcal{V}} \pi_i^*(t_s) = \rho_d(t_s) + \alpha$, which can be made arbitrarily smaller/larger than $\rho_d(t_s)$.
\end{itemize}
\begin{remark}
    Although we consider drift in the states $x_i$ by the malicious agents, a similar deviation in the $y_i$ updates can be considered. However, the malicious DER can cause any desired result by only manipulating $x_i$ states and doesn't have additional benefits by manipulation of the $y_i$ states. Further, the intruder model with states $\hat{x}_m$ and the update~\eqref{eq:appr_update_num_mal} allows for modeling DERs that are not necessarily malicious but are misbehaving (and send imperfect values to neighboring DERs) due to errors in computation or communication or any other internal failures.
\end{remark}

\begin{figure*}[b]
	\centering
	\subfloat[]{\includegraphics[scale=0.28,trim={0.39cm 0.5cm 5cm 0.48cm},clip]{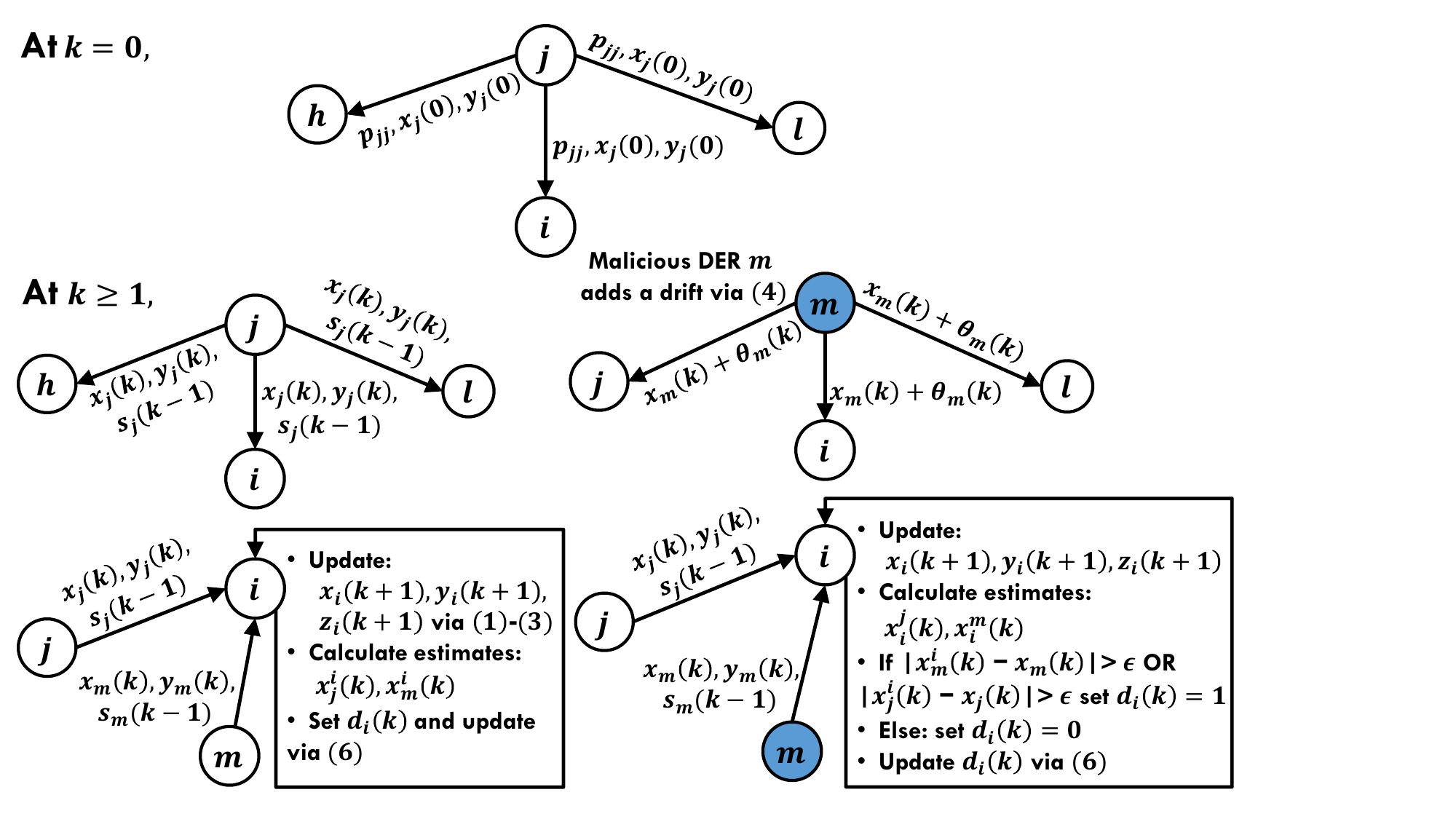}%
	\label{fig:detection_scheme}}
	\subfloat[]{\includegraphics[scale=0.32,trim={3.3cm 0cm 0cm 0cm},clip]{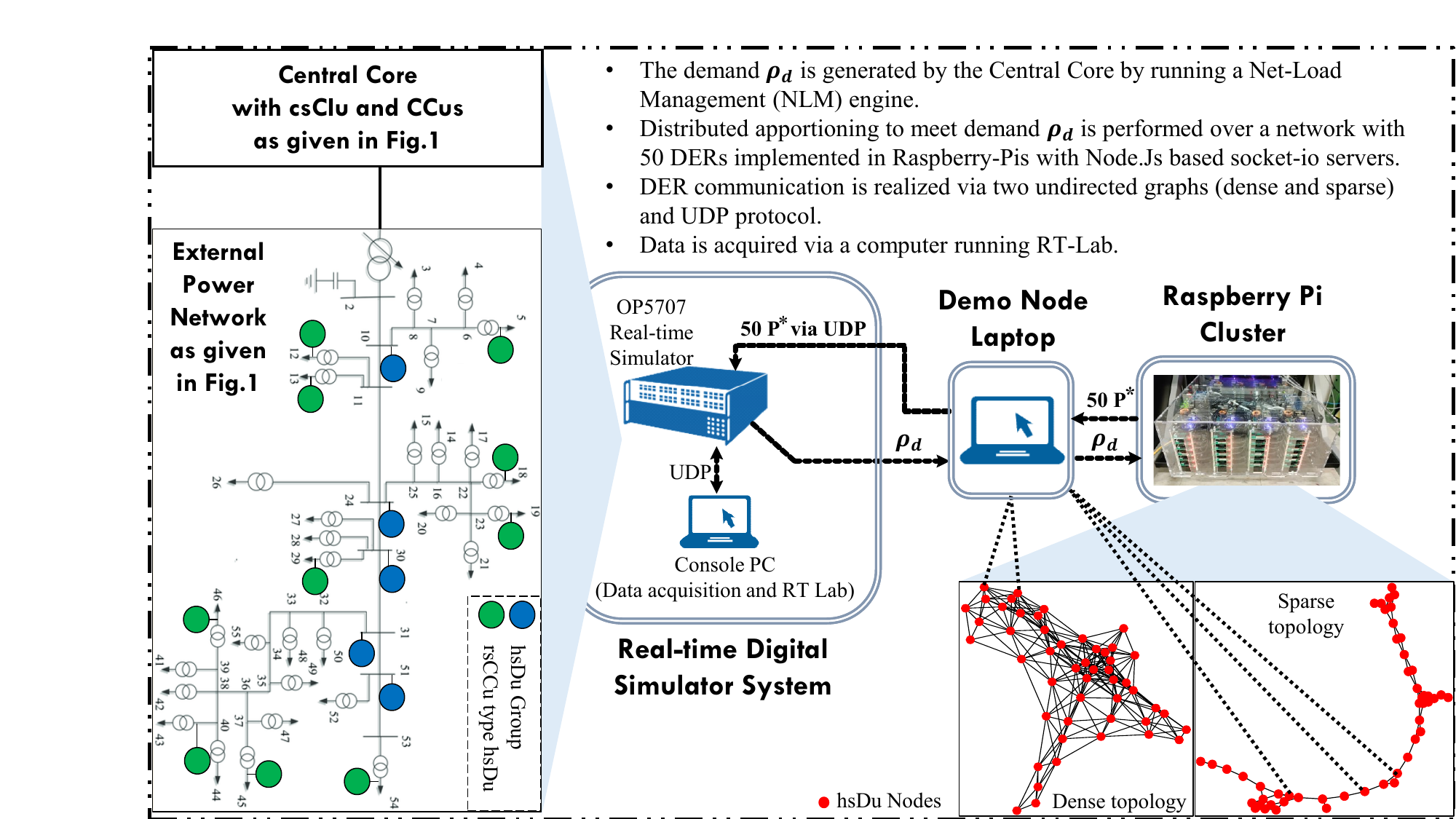}%
	\label{fig:chil_setup}}
	\caption{(a) Resilient distributed power apportioning scheme with malicious DER $m$. Honest DER $i$ detects the presence of $m$, (b) Laboratory CHIL setup.}
	\label{fig:detection_and_chil}
\end{figure*}
\subsection{Resilient Distributed Power Apportioning}
In light of the malicious DERs, we strive to equip each DER with the ability to detect deviant behavior and isolate the communication links of the malicious DERs. We make the following assumption:
\begin{assump}\label{assmp:trust_beginning}
    At the beginning of the power apportioning, i.e. $k = 0$, the information sent by the DERs is truthful.
\end{assump}
\noindent Next, we describe this strategy: Given $\mathcal{G}=\{\mathcal{V}, \mathcal{E}\}$, each DER maintains additional states $x_j^i(k)$ for all its neighbors $j \in \mathcal{N}_i^-$ and a \df $d_i$ utilized by DER $i$ to indicate the presence of a malicious DER in its neighborhood. The distributed apportioning protocol is modified to allow each DER $i$ to utilize these additional states $x_j^i$. The augmented distributed apportioning protocol has the following key steps:\\
$\bullet$ At the start of the apportioning protocol, i.e. $k = 0$,\\
\hspace*{0.05in} $\Diamond$ each DER $i$, in addition to $x_j(0), y_j(0)$, requires all its in-neighbors  $j \in \mathcal{N}_i^-$ to send their weights $p_{jj}$. \\   
$\bullet$ At any subsequent iteration $k \geq 1$,\\
\hspace*{0.05in} $\Diamond$ every DER $j$ sends $x_j(k), y_j(k)$ and an additional state $s_j(k-1) = \sum_{l \in \mathcal{N}_j^-} p_{jl} x_l(k-1)$ to its out-neighbors $i \in \mathcal{N}_j^+$. \\
\hspace*{0.05in} $\Diamond$ every DER $i$ then updates its states to $x_i(k+1), y_i(k+1), z_i(k+1)$ via the updates~\eqref{eq:appr_update_num}-\eqref{eq:appr_update_ratio}, and the states $x_j^i$ as $x_j^i(k) = p_{jj} x_j(k-1) + s_j(k-1)$.\\
\hspace*{0.05in} $\Diamond$ if there is a mismatch between the estimated state $x_j^i(k)$ and the actual received state $x_j(k)$ then the agent $i$ detects that the communication links between its neighbor $j$ and itself are compromised. In particular, with small $\varepsilon > 0$, if $| x_j^i(k) - x_j(k)| > \varepsilon$ then agent $i$ detects node $j$ is compromised. Further, DER $i$ sets its $\mathrm{detection\_flag}$ $d_i$ to $1$ and adds the communication address of DER $j$ to the list of compromised communication links, $\mathcal{C}^m_i$, in its neighborhood. \\
\hspace*{0.05in} $\Diamond$ Any DER $i$ that detects a malicious DER in its neighborhood, propagates its \df in the network. This is achieved via a $1$-bit consensus algorithm. Each DER initializes $d_i(k)$ at every $k = l\mathcal{D}$, where $\mathcal{D}$ is an upper bound on the diameter of the graph $\mathcal{G}(\mathcal{V}, \mathcal{E})$, iteration for $l \in \{0,1,2,\dots\}$ with $1$ or $0$ depending on whether the DER has detected a malicious DER in its neighborhood or not, and updates its \df value every iteration using,
    \begin{align}\label{eq:bitCons}
        d_i(k+1) = \textstyle \bigcup_{j\in \mathcal{N}_j^- \cup \{i\}} d_j(k),
    \end{align}
where $\bigcup$ denotes the ``OR" operation and $d_j(0)=1$ if node $j$ has detected a malicious DER at initialization instant $0$ and $d_j(0)=0$ otherwise. Clearly, if $d_j(0)=1$ for any DER $j \in \mathcal{V}$, then $ d_i(\mathcal{D})=1$ for all $i \in \mathcal{V}$ where $\mathcal{D}$ is the diameter. Thus, each DER can use $d_i(\mathcal{D})$ as a criterion to ascertain the presence of a malicious DER in the network.
The modified power apportioning algorithm is presented in Fig.~\ref{fig:detection_and_chil}\subref{fig:detection_scheme}. We present the following result for the detection scheme. 
\begin{theorem}
    Let Assumption~\ref{assmp:trust_beginning} hold. Then the detection scheme detailed above and in Fig.~\ref{fig:detection_and_chil}\subref{fig:detection_scheme} leads to detecting any malicious DER obscuring its true states via~\eqref{eq:intruder_model}.
\end{theorem}
\begin{proof}
We present an argument by considering a particular DER index $i$, the argument holds symmetrically for all DERs $i \in \mathcal{V}$. Following the procedure detailed in Fig.~\ref{fig:detection_and_chil}\subref{fig:detection_scheme} at the starting iteration $k = 0$ node $i$ receives $x_j(0), y_j(0), p_{jj}$ from all its in-neighboring DERs. Note that under Assumption~\ref{assmp:trust_beginning} $x_j(0), y_j(0)$ and $p_{jj}$ are the true values of these quantities. Following this, DER $i$ updates its states following~\eqref{eq:appr_update_num}-\eqref{eq:appr_update_den}, $x_i(1) =  p_{ii}x_i(0)+\sum_{j\in \mathcal{N}_i^-}p_{ij}x_j(0)$, $y_i(1) = p_{ii}y_i(0)+\sum_{j\in \mathcal{N}_i^-}p_{ij}y_j(0)$, $z_i(1) = \frac{x_i(1)}{y_i(1)}$. At $k = 1$, the DER $i$ sends $x_i(1), y_i(1)$ and $s_i(0) = \sum_{l\in \mathcal{N}_i^-}p_{il}x_l(0)$ to its out-neighboring DERs and receives the $x_j(1), y_j(1)$ and $s_j(0) = \sum_{l\in \mathcal{N}_j^-}p_{jl}x_l(0)$ from its in-neighboring DERs $j$. Following this DER $i$ determines the estimates $x_j^i(1) = p_{jj} x_j(0) + s_j(0)$. Without loss of generality assume DER $j$ behaves maliciously at iteration $k=1$, and adds a linear drift $\theta_j(1)$ to the true state $x_j(1)$ and sends $\hat{x}_j(1) = x_j(1) + \theta_j(1)$ to its out-neighbors including DER $i$. Thus at DER $i$, given a samll detection threshold $\varepsilon$, $|\hat{x}_j(1) - x_j^i(1)| = |p_{jj} x_j(0) + s_j(0) - x_j(1) - \theta_j(1)| = |p_{jj} x_j(0) + s_j(0) - p_{jj}x_j(0) - s_j(0) - \theta_j(1)| = |\theta_j(1)| > \varepsilon$ and $j$ gets detected as a maliciously behaving DER by DER $i$. Similarly, all other DERs in the network compute the difference between the states sent by their in-neighbors and the estimates maintained locally to detect any malicious DER in their neighborhood. Thus, any linear drift $\theta_j$ aimed to shift the DER estimates significantly away from the true state values can be detected. This completes the proof.  
\end{proof}
After all the DERs ascertain the presence of malicious DERs via the \df \hspace{-0.1cm}, a new instance of the distributed apportioning algorithm is started among the trustworthy DERs by removing the compromised communication links $C_i^m \in \mathcal{E}$ from their neighborhoods.

\section{CHIL Demonstration and Results}\label{sec:results}
\subsection{Experimental Configuration}

\begin{figure*}[t]
	\centering
	\subfloat{\includegraphics[scale=0.284,trim={0.5cm 0.5cm 1cm 1.9cm},clip]{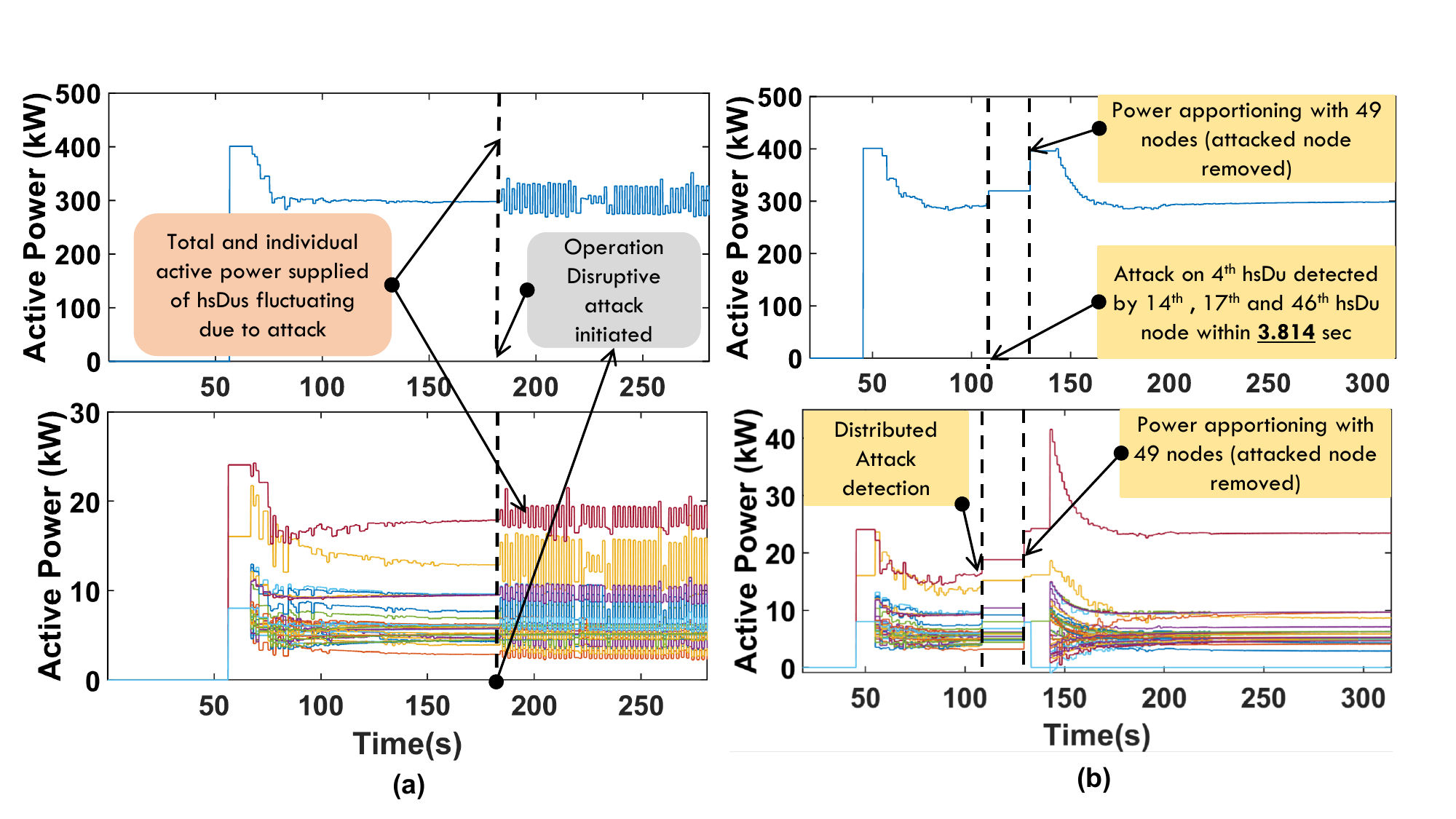}}%
	\subfloat{\includegraphics[scale=0.284,trim={0.5cm 0.5cm 1cm 1.9cm},clip]{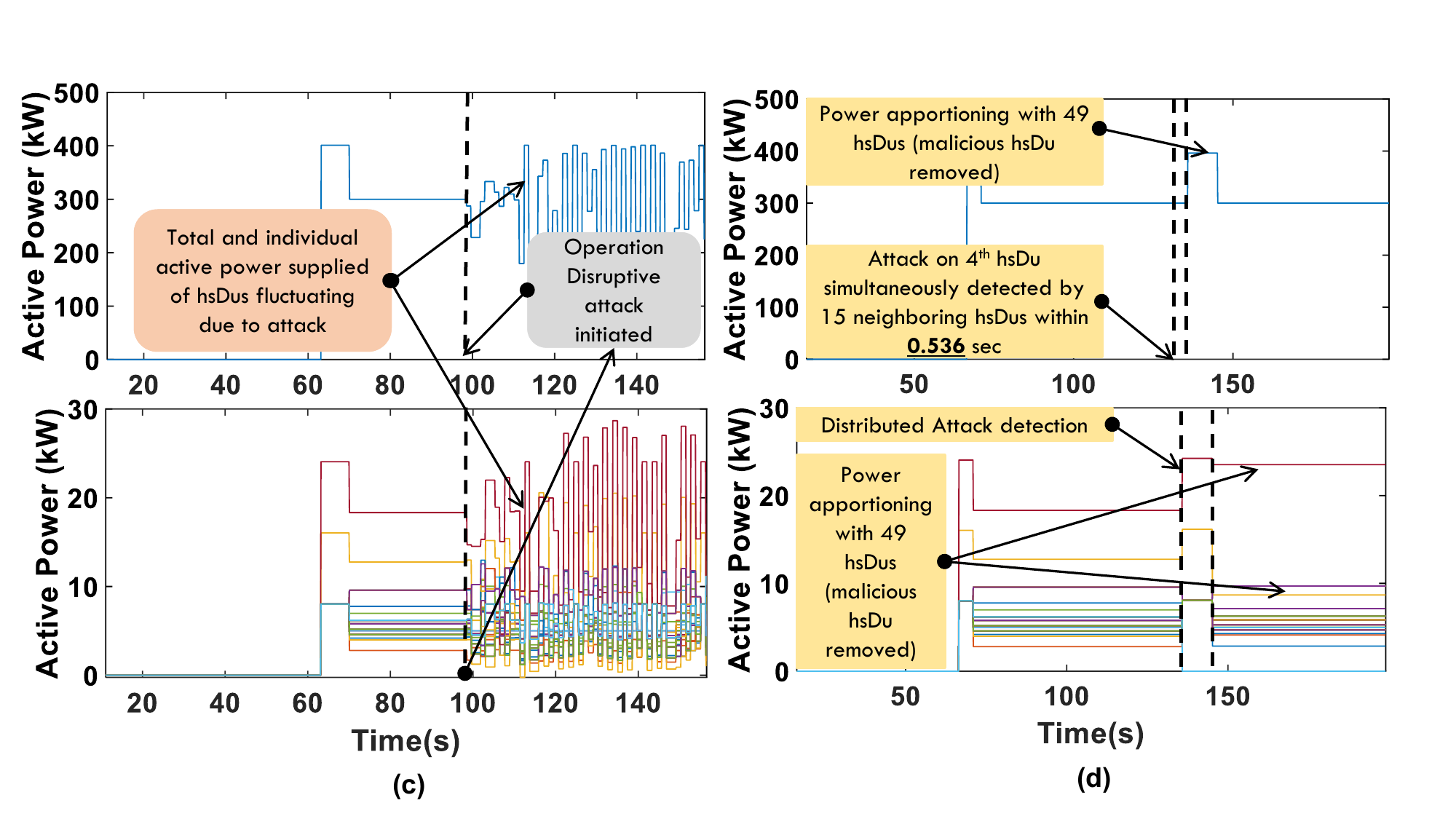}}\\ \vspace{-0.1in}
    \subfloat{\includegraphics[scale=0.284,trim={0.5cm 0.5cm 1cm 1.9cm},clip]{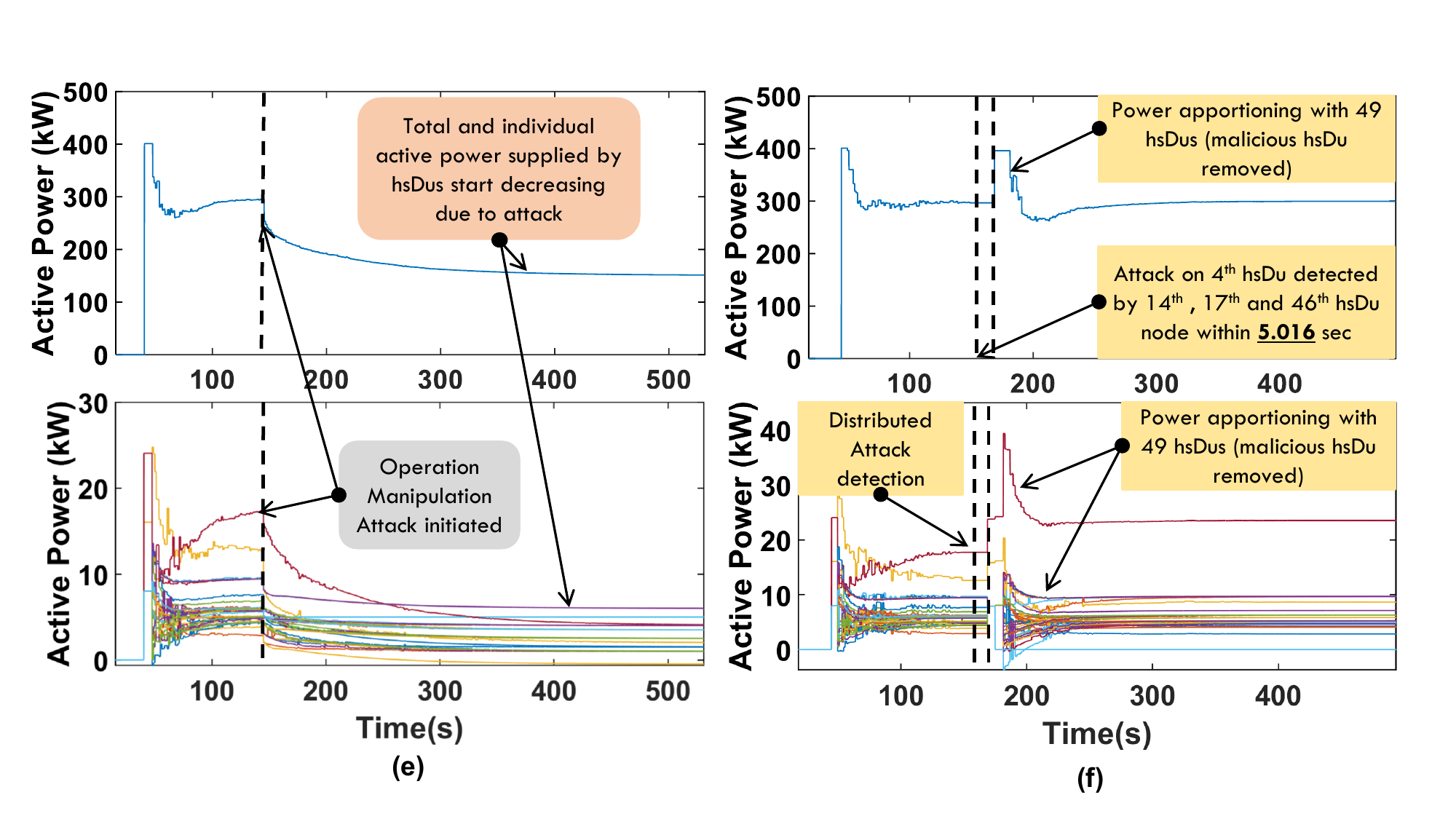}}%
	\subfloat{\includegraphics[scale=0.28,trim={0.5cm 0.5cm 1cm 1.9cm},clip]{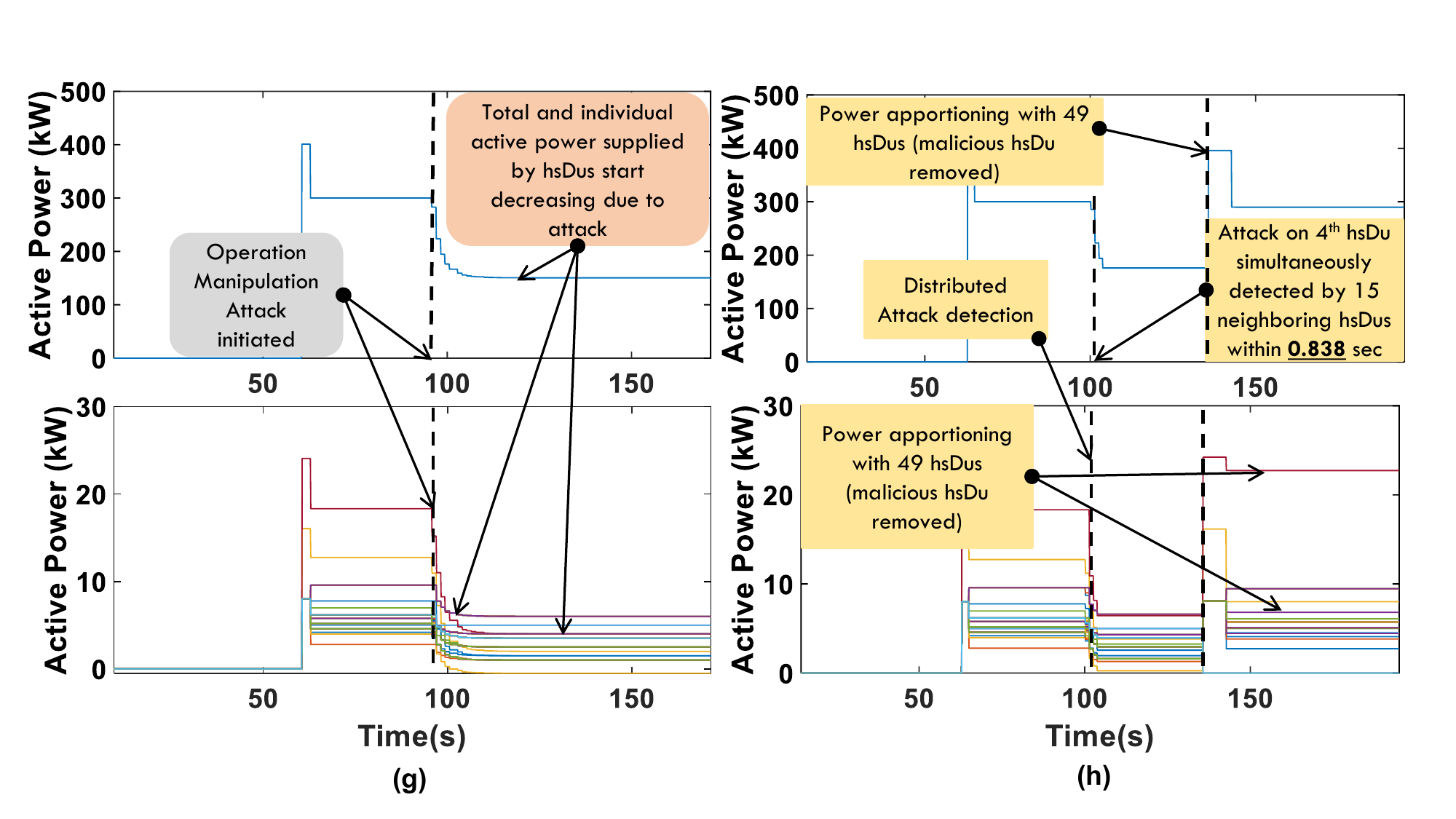}}%
	\caption{Performance of the distributed power apportioning scheme with malicious hsDu detection under $\mathtt{ATTACK}$-$\mathtt{1}$ (a)-(d) and-$\mathtt{ATTACK}$-$\mathtt{2}$ (e)-(h).}
	\label{fig:combined_results}
\end{figure*}


Fig.~\ref{fig:detection_and_chil}\subref{fig:chil_setup} shows the entire CHIL experimental setup. There are three important components in the experimental setup:\\
\textbf{1.} The microgrid power network under study consisting of the Central Core and the external power network as shown in Fig.~\ref{fig:microgrid} is emulated using the multi-timescale real-time simulation method, with time-step of $T_\mathrm{s}$ = $50\mu$s and $T_\mathrm{ss}$ = $5$ms via eMEGASIM and ePHASORSIM platform respectively inside the OP$5700$ RT-simulator (RTS) manufactured by OPAL-RT.\\
\textbf{2.} The central core is equipped with a central dispatch control system that provides high-speed, cost-optimal coordination of net loads (a combination of generation and controllable loads) \cite{nlm1,nlm2}. We refer to this central dispatch controller as the NLM engine. The NLM runs on a desktop computer with $16$ GB RAM and an Intel Core i$7$ processor running at $1.90$ GHz, utilizing  Python $3.7.1$. The devices in the central core of the emulated microgrid communicate with the NLM engine via using standard User Datagram Protocol (UDP) \cite{udp}.\\
\textbf{3.} The resilient distributed power apportioning algorithm with $50$ nodes connected via pre-defined topology information, embedded in a JavaScript Object Notation (JSON) file, running in a Raspberry-based cluster device. The NLM engine transmits the total power demand command $\rho_d$ to the power apportioning engine via UDP and the power set point for all the $50$ DERs, situated in the external electrical network of the emulated microgrid, are transmitted back to the RT simulator via UDP after the apportioning algorithm has converged.
\par Realistic load profiles (obtained through the metering system by the energy management group of the University of Minnesota), and the real-world solar irradiance profiles (obtained from the National Renewable Energy Laboratory’s Solar Measurement and Instrumentation Data Center and the National Solar Radiation Database \cite{nrel}) are fed into the real-time simulator with $1$ sec resolution.
\par We consider two attack scenarios where hsDu $4$ is malicious and has two different objectives. For both attack scenarios, the total power demand requested by the central core is $300$ kW. Further, we validate the performance of the proposed resilient power apportioning scheme with two extreme cases of communication topology, as shown in Fig.~\ref{fig:microgrid_communication}\subref{fig:communication}(i) and Fig.~\ref{fig:microgrid_communication}\subref{fig:communication}(ii), between the hsDus: one sparse (here hsDu $14,17$ and $46$ are the neighbors of hsDu $4$) and another dense network (here hsDu $4$ has $15$ neighboring hsDus). The convergence of updates~\eqref{eq:appr_update_num} and~\eqref{eq:appr_update_ratio} to the solution~\eqref{eq:z_star} depends on the connectivity of the underlying graph $\mathcal{G}$. Therefore, the sparse and dense topologies with very different connectivities provide a comparison of the performance of the developed resilient distributed apportioning algorithm under practical conditions. The two attack scenarios are as follows:\\
$\bullet~\mathtt{ATTACK}$-$\mathtt{1}$ Operation Disruption Attack: The malicious hsDu $4$ injects false state estimates in the power apportioning scheme such that the algorithm does not converge and no power ancillary support to the central core can be provided. To achieve this hsDu $4$ chooses a divergent perturbation sequence $\theta_4(k) = \delta$ if $k$ is even and $\theta_4(k) = -\delta$ if $k$ is odd.\\
$\bullet~\mathtt{ATTACK}$-$\mathtt{2}$ Operation Manipulation Attack: The malicious hsDu $4$ injects false state estimates to alter the steady-state power dispatch set-points of the honest hsDus to provide only $50\%$, $150$ kW, of the requested ancillary demand.  
\subsection{Results and Discussions}
During the CHIL experiments for both attack scenarios, we perform \textbf{two} runs of the distributed power apportioning engine: (i) without and (ii) with the distributed detection scheme. We report the results without the detection scheme first to show the impact of the attack on the power network and then present the performance of the apportioning algorithm with the detection scheme. Figs.~\ref{fig:combined_results}(a)-(h) provide the results of the CHIL experiments under the attack scenarios $\mathtt{ATTACK}$s-$\mathtt{1}$ and-$\mathtt{2}$. In all the Figs. we present the total power output of the hsDu network in the top half and the power output of individual hsDus in the bottom half.\\
\hspace*{0.1in} Figs.~\ref{fig:combined_results}(a)-(b) and Figs.~\ref{fig:combined_results}(c)-(d) present the results under $\mathtt{ATTACK}$-$\mathtt{1}$ for the sparse and the dense communication topology respectively. It can be seen that the power set-points of the hsDus do not reach a steady state and fluctuate a lot (Figs.~\ref{fig:combined_results}(a),(c)) and potentially can make the system unstable if no countermeasures are taken by the honest hsDu towards detection of the malicious hsDu $4$. With the proposed detection scheme the hsDus in the neighborhood of hsDu $4$ detects the operation disruption attack within $3.814$ secs and $0.536$ secs in the sparse and dense communication topology respectively (Figs.~\ref{fig:combined_results}(b),(d)). Following the detection of the malicious hsDu $4$, the remaining $49$ hsDus restart power apportioning to provide the requested demand of $300$ kW to the central core.\\
\hspace*{0.1in} Figs.~\ref{fig:combined_results}(e)-(f) and Figs.~\ref{fig:combined_results}(g)-(h) present the results under $\mathtt{ATTACK}$-$\mathtt{2}$ for the sparse and the dense communication topology respectively. Although the honest hsDus can meet the desired demand of $300$ kW (Table~\ref{tab:rating2}). The hsDu $4$ is able to manipulate the convergence of the power apportioning algorithm reducing the ancillary power support to the central core to $50 \%$ ($150$ kW), Figs.~\ref{fig:combined_results}(e),(g). Thus, a mechanism to detect the malicious actions of the hsDus is necessitated. With the proposed detection scheme the hsDus in the neighborhood of hsDu $4$ detects the operation manipulation attack within $5.016$ secs and $0.838$ secs in the sparse and dense communication topology respectively (Figs.~\ref{fig:combined_results}(b),(d)). Following the detection of the malicious hsDu $4$, the remaining $49$ hsDus restart the power apportioning to provide the requested power demand of $300$ kW to the central core. We remark that the time of detection of the attack is smaller in the dense communication topology as in the dense graph there are more simultaneous detections due to a higher degree of the hsDus.\\
\hspace*{0.1in} Thus, the augmentation of the proposed malicious node detection with the power apportioning leads to a resilient operation of microgrids allowing the hsDus to aggregate their resources in a robust manner to meet the ancillary power demand of the central core.

\section{Conclusion and Future Work}\label{sec:conclusion}
In this article, we considered resilient distributed aggregation of DERs for ancillary service dispatch in 
 the power networks. We developed a distributed detection scheme that allows the DERs to detect and isolate the communication links of the maliciously behaving DERs. We showed that the proposed distributed scheme leads to the detection of the malicious DERs. We tested the performance of the proposed algorithm under two different attack scenarios in the CHIL experiments. With the proposed scheme the DERs detect the malicious behavior within seconds in both the attack scenarios. 
 In the CHIL experiment study, we considered two attack scenarios. An extension of the current work is to test the performance of the developed distributed resilient power apportioning scheme under more attack scenarios. 
\bibliography{references}

\end{document}